\documentclass[conference]{IEEEtran}

\usepackage{cite}
\usepackage{graphicx,amssymb,amstext,amsmath,equation,dsfont}
\usepackage{bm}
\usepackage{setspace}
\usepackage{stfloats}
\usepackage{flushend}
\usepackage{url}
\usepackage{array}
\usepackage{mdwmath}
\usepackage{mdwtab}
\usepackage{epsfig}
\usepackage{fixltx2e}
\usepackage{booktabs}
\usepackage[EULERGREEK]{sansmath}
\usepackage{pstricks}
\usepackage{exscale}
\usepackage[T1]{fontenc} 
\usepackage{psfrag}
\usepackage{subfig}

\newtheorem{theorem}{Theorem}

\newtheorem{corollary}{Corollary}
\newtheorem{remark}{Remark}

\newcommand{\yv}{{\bm y}}

\newcommand{\yrv}{\bm Y}

\newcommand{\xv}{{\bm x}}

\newcommand{\xrv}{\bm X}

\newcommand{\zrv}{\bm Z}

\newcommand{\HM}{{\mathsf H}}

\newcommand{\HRM}{\mathbb{H}}

\newcommand{\ERM}{\mathbb{E}}

\newcommand{\nt}{{n_{\rm t}}}
\newcommand{\nr}{{n_{\rm r}}}
\newcommand{\nts}{{n_{{\rm t},s}}}
\newcommand{\ntone}{{n_{{\rm t},1}}}
\newcommand{\nttwo}{{n_{{\rm t},2}}}
\newcommand{\SNR}{{\sf SNR}}

\newcommand{\gmi}{I^{\rm gmi}}

\newcommand{\trans}[1]{#1^{\textnormal{\textsf{\tiny T}}}} 
\ifCLASSINFOpdf
\else
\fi

\DeclareMathAlphabet{\mathpzc}{OT1}{pzc}{m}{it}

\DeclareSymbolFont{lettersA}{U}{txmia}{m}{it}
 \DeclareMathSymbol{\real}{\mathord}{lettersA}{"92}
 \DeclareMathSymbol{\field}{\mathord}{lettersA}{"83} 
 \DeclareMathSymbol{\integ}{\mathord}{lettersA}{"9A}

\newcommand{\hermi}[1]{#1^{\dagger}} 

\allowdisplaybreaks[4]

\begin{document}
\title{Nearest Neighbour Decoding with \\Pilot-Assisted Channel Estimation for \\Fading Multiple-Access Channels}

\author{\IEEEauthorblockN{A. Taufiq Asyhari, Tobias Koch and Albert Guill{\'e}n i F{\`a}bregas}
\IEEEauthorblockA{University of Cambridge, Cambridge CB2 1PZ, UK\\Email:  taufiq-a@ieee.org, tobi.koch@eng.cam.ac.uk, guillen@ieee.org}}


\maketitle

\begin{abstract}
We study a noncoherent multiple-input multiple-output (MIMO) fading multiple-access channel (MAC), where the transmitters and the receiver are aware of the statistics of the fading, but not of its realisation. We analyse the rate region that is achievable with nearest neighbour decoding and pilot-assisted channel estimation and determine the corresponding pre-log region, which is defined as the limiting ratio of the rate region to the logarithm of the signal-to-noise ratio (SNR) as the SNR tends to infinity.
\renewcommand{\thefootnote}{}
\footnotetext{The work of A. T. Asyhari has been partly supported by the Yousef Jameel Scholarship at University of Cambridge. The work of T. Koch has received funding from the European's Seventh Framework Programme (FP7/2007--2013) under grant agreement No.\ 252663.}
\setcounter{footnote}{0}
\end{abstract}



\IEEEpeerreviewmaketitle

\graphicspath{{Figure/EPS/}{Figure/}}

\section{Introduction and Channel Model}
\label{sec:intro}

We study a two-user multiple-input multiple-output (MIMO) fading multiple-access channel (MAC), where two terminals wish to communicate with a third one, and where the channels between the terminals are MIMO fading channels. We consider a noncoherent channel model, where all terminals are aware of the statistics of the fading, but not of its realisation. We are interested in the achievable-rate region that can be achieved with nearest neighbour decoding and pilot-assisted channel estimation. We focus on the high signal-to-noise ratio (SNR) regime. In particular, we study the pre-log region, defined as the limiting ratio of the achievable-rate region to $\log\SNR$ as the $\SNR$ tends to infinity.

The pre-log of point-to-point MIMO fading channels achievable with nearest neighbour decoding and pilot-assisted channel estimation was studied in \cite{IEEE:asyhari_etal:nearest_neighbour_ISIT11}. It was demonstrated that it coincides with the capacity pre-log (defined as the limiting ratio of \emph{capacity} to $\log\SNR$ as the $\SNR$ tends to infinity) for multiple-input single-output (MISO) fading channels, derived by Koch and Lapidoth \cite{IEEE:koch:fadingnumber_degreeoffreedom}, and that it achieves the best so far known lower bound on the pre-log of MIMO fading channels, derived by Etkin and Tse \cite{IEEE:etkin:degreeofffreedomMIMO}.

In this paper, we extend the analysis in \cite{IEEE:asyhari_etal:nearest_neighbour_ISIT11} to the two-user MIMO fading MAC where the first user has $\ntone$ antennas, the second user has $\nttwo$ antennas and the receiver has $\nr$ antennas. The channel model is depicted in Fig.~\ref{fig:mac-system-model}. The channel output at time instant $k \in \integ$ (where $\integ$ denotes the set of integers) is a complex-valued $\nr$-dimensional random vector given by 
\begin{equation}
\label{eq:channel}
 \yrv_k = \sqrt{{\sf SNR}}\, \HRM_{1,k} \xv_{1,k} + \sqrt{{\sf SNR}}\, \HRM_{2,k} \xv_{2,k} + \zrv_k.
\end{equation}
Here  $\xv_{s,k} \in \field^{\nts}$ denotes the time-$k$ channel input vector corresponding to user $s$, $s=1,2$ (with $\field$ denoting the set of complex numbers); $\HRM_{s,k} \in \field^{\nr \times \nts}$ denotes the fading matrix at time $k$ corresponding to user $s$, $s=1,2$; $\SNR$ denotes the average SNR for each transmit antenna; and $\zrv_k \in \field^{\nr}$ denotes the additive noise vector at time $k$. 

\begin{figure}[t]
\begin{center}
\begin{psfrags}
 \psfrag{a}[c]{Tx}
 \psfrag{b}[c]{\footnotesize $s=1$}
 \psfrag{c}[c]{\footnotesize $s=2$}
 \psfrag{d}[c]{Rx}
 \psfrag{e}[c]{$m_1$}
 \psfrag{f}[c]{$m_2$}
 \psfrag{g}[c]{$\left(\hat m_1, \hat m_2 \right)$}
 \includegraphics[width=0.43\textwidth]{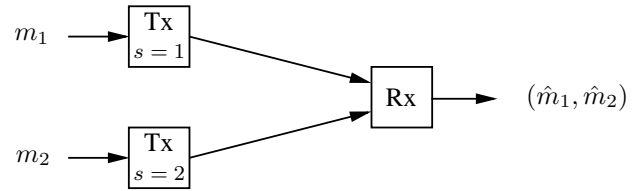}
 \vspace*{-2mm}
\end{psfrags}
\end{center}
\caption{The two-user MAC system model.}
\label{fig:mac-system-model}
\vspace*{-5mm}
\end{figure}

The noise process $\{\zrv_k, k \in \integ \}$ is a sequence of independent and identically distributed (i.i.d.) complex Gaussian random vectors with zero mean and covariance matrix ${\sf I}_{\nr}$, where ${\sf I}_{\nr}$ is the $\nr \times \nr$ identity matrix. 

The fading processes $\{ \HRM_{s,k}, k \in \integ \}$, $s = 1,2$ are stationary, ergodic and Gaussian. We assume that the $(\ntone \cdot\nr + \nttwo \cdot \nr)$ processes $\{H_{s,k}(r,t), k \in \integ \}$, $s=1,2$, $r=1,\ldots,\nr$, $t=1,\ldots,\nts$ are independent and have the same law, with each process having zero-mean, unit-variance and power spectral density $f_H (\lambda)$, $-\frac{1}{2} \leq \lambda \leq \frac{1}{2}$. Thus, $f_H(\cdot)$ is a non-negative function satisfying
 \begin{equation}
 \mathsf{E} \left[ H_{s,k+m} (r,t) \hermi{H}_{s,k} (r,t) \right] = \int^{1/2}_{-1/2} e^{i2\pi m \lambda} f_H (\lambda) d \lambda
\end{equation}
where $\hermi{(\cdot)}$ denotes complex conjugation. We further assume that the power spectral density $f_H(\cdot)$ has bandwidth $\lambda_D<1/2$, i.e., $f_H (\lambda) = 0$ for $|\lambda| > \lambda_D$ and $f_H(\lambda)>0$ otherwise.

We finally assume that the fading processes $\{\HRM_{s,k}, k\in\integ\}$, $s=1,2$ and the noise process $\{\zrv_k, k \in \integ \}$ are independent and that their joint law does not depend on $\{\xv_{s,k}, k\in\integ\}$, $s=1,2$. We consider a noncoherent channel model, where the transmitters and the receiver are aware of the statistics of $\{\HRM_{s,k}, k\in\integ\}$, $s=1,2$, but not of their realisations.

\section{Transmission Scheme}
\label{sec:trans_scheme}
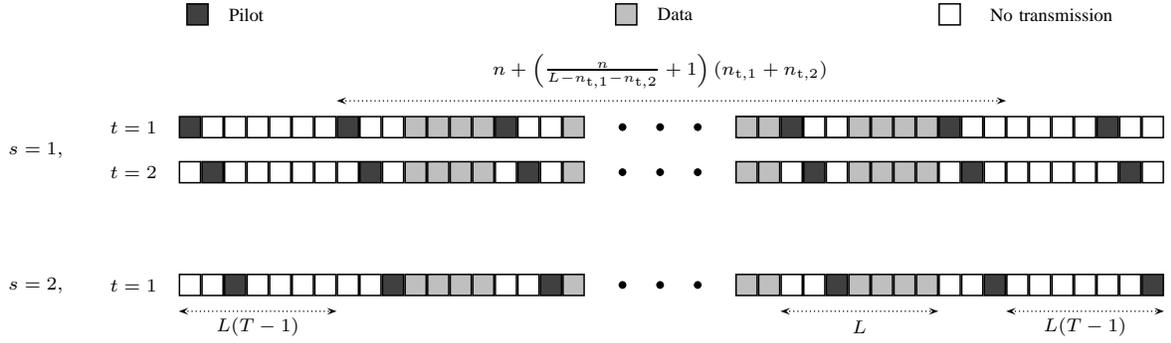
\begin{figure*}[t]
 \begin{center}
\begin{pspicture}(-5,-2.65)(12,1.8)

\psframe[linewidth=0.02,fillstyle=solid,fillcolor=darkgray](-2,1.5)(-1.7,1.8)
\rput(-1.2,1.65){\scriptsize Pilot}

\psframe[linewidth=0.02,fillstyle=solid,fillcolor=lightgray](3.7,1.5)(4,1.8)
\rput(4.5,1.65){\scriptsize Data}

\psframe[linewidth=0.02](8.0,1.5)(8.3,1.8)
\rput(9.5,1.65){\scriptsize No transmission}

\rput(4.3,0.9){\scriptsize $n + \left(\frac{n}{L - \ntone - \nttwo} + 1 \right) (\ntone + \nttwo)$}
\psline[linewidth=0.02,linestyle=dotted,dotsep=1pt]{<->}(0,0.5)(8.9,0.5)

\rput(6.95,-2.5){\scriptsize $L$}
\psline[linewidth=0.02,linestyle=dotted,dotsep=1pt]{<->}(5.9,-2.3)(8,-2.3)

\rput(9.95,-2.5){\scriptsize $L(T-1)$}
\psline[linewidth=0.02,linestyle=dotted,dotsep=1pt]{<->}(8.9,-2.3)(11,-2.3)

\rput(-1.05,-2.5){\scriptsize $L(T-1)$}
\psline[linewidth=0.02,linestyle=dotted,dotsep=1pt]{<->}(-2.1,-2.3)(0,-2.3)

\rput(-4,-0.15){\scriptsize $s=1,$}

{
\rput(-2.7,0.15){\scriptsize $t=1$}

\psframe[linewidth=0.02,fillstyle=solid,fillcolor=darkgray](-2.1,0)(-1.8,0.3)
\psframe[linewidth=0.02](-1.8,0)(-1.5,0.3)
\psframe[linewidth=0.02](-1.5,0)(-1.2,0.3)
\psframe[linewidth=0.02](-1.2,0)(-0.9,0.3)
\psframe[linewidth=0.02](-0.9,0)(-0.6,0.3)
\psframe[linewidth=0.02](-0.6,0)(-0.3,0.3)
\psframe[linewidth=0.02](-0.3,0)(0,0.3)

\psframe[linewidth=0.02,fillstyle=solid,fillcolor=darkgray](0,0)(0.3,0.3)
\psframe[linewidth=0.02](0.3,0)(0.6,0.3) 
\psframe[linewidth=0.02](0.6,0)(0.9,0.3)
\psframe[linewidth=0.02,fillstyle=solid,fillcolor=lightgray](0.9,0)(1.2,0.3)
\psframe[linewidth=0.02,fillstyle=solid,fillcolor=lightgray](1.2,0)(1.5,0.3)
\psframe[linewidth=0.02,fillstyle=solid,fillcolor=lightgray](1.5,0)(1.8,0.3)
\psframe[linewidth=0.02,fillstyle=solid,fillcolor=lightgray](1.8,0)(2.1,0.3)

\psframe[linewidth=0.02,fillstyle=solid,fillcolor=darkgray](2.1,0)(2.4,0.3)
\psframe[linewidth=0.02](2.4,0)(2.7,0.3) 
\psframe[linewidth=0.02](2.7,0)(3,0.3)
\psframe[linewidth=0.02,fillstyle=solid,fillcolor=lightgray](3,0)(3.3,0.3)

\pscircle[fillstyle=solid,fillcolor=black,linewidth=0.0875,linecolor=white](3.8,0.15){0.14}
\pscircle[fillstyle=solid,fillcolor=black,linewidth=0.0875,linecolor=white](4.3,0.15){0.14}
\pscircle[fillstyle=solid,fillcolor=black,linewidth=0.0875,linecolor=white](4.8,0.15){0.14}

\psframe[linewidth=0.02](5.3,0)(8.6,0.3)
\psframe[linewidth=0.02,fillstyle=solid,fillcolor=lightgray](5.3,0)(5.6,0.3)
\psframe[linewidth=0.02,fillstyle=solid,fillcolor=lightgray](5.6,0)(5.9,0.3)
\psframe[linewidth=0.02,fillstyle=solid,fillcolor=darkgray](5.9,0)(6.2,0.3)
\psframe[linewidth=0.02](6.2,0)(6.5,0.3)
\psframe[linewidth=0.02](6.5,0)(6.8,0.3)
\psframe[linewidth=0.02,fillstyle=solid,fillcolor=lightgray](6.8,0)(7.1,0.3)
\psframe[linewidth=0.02,fillstyle=solid,fillcolor=lightgray](7.1,0)(7.4,0.3)
\psframe[linewidth=0.02,fillstyle=solid,fillcolor=lightgray](7.4,0)(7.7,0.3)
\psframe[linewidth=0.02,fillstyle=solid,fillcolor=lightgray](7.7,0)(8,0.3)

\psframe[linewidth=0.02,fillstyle=solid,fillcolor=darkgray](8,0)(8.3,0.3)
\psframe[linewidth=0.02](8.3,0)(8.6,0.3)
\psframe[linewidth=0.02](8.6,0)(8.9,0.3)
\psframe[linewidth=0.02](8.9,0)(9.2,0.3)
\psframe[linewidth=0.02](9.2,0)(9.5,0.3)
\psframe[linewidth=0.02](9.5,0)(9.8,0.3)
\psframe[linewidth=0.02](9.8,0)(10.1,0.3)

\psframe[linewidth=0.02,fillstyle=solid,fillcolor=darkgray](10.1,0)(10.4,0.3)
\psframe[linewidth=0.02](10.4,0)(10.7,0.3)
\psframe[linewidth=0.02](10.7,0)(11,0.3)

}

{
\rput(-2.7,-0.45){\scriptsize $t=2$}

\psframe[linewidth=0.02](-2.1,-0.6)(-1.8,-0.3)
\psframe[linewidth=0.02,fillstyle=solid,fillcolor=darkgray](-1.5,-0.6)(-1.8,-0.3)
\psframe[linewidth=0.02](-1.5,-0.6)(-1.2,-0.3) 
\psframe[linewidth=0.02](-1.2,-0.6)(-0.9,-0.3)
\psframe[linewidth=0.02](-0.9,-0.6)(-0.6,-0.3) 
\psframe[linewidth=0.02](-0.6,-0.6)(-0.3,-0.3)
\psframe[linewidth=0.02](-0.3,-0.6)(0,-0.3) 

\psframe[linewidth=0.02](0,-0.6)(0.3,-0.3)
\psframe[linewidth=0.02,fillstyle=solid,fillcolor=darkgray](0.3,-0.6)(0.6,-0.3)
\psframe[linewidth=0.02](0.6,-0.6)(0.9,-0.3)
\psframe[linewidth=0.02,fillstyle=solid,fillcolor=lightgray](0.9,-0.6)(1.2,-0.3)
\psframe[linewidth=0.02,fillstyle=solid,fillcolor=lightgray](1.2,-0.6)(1.5,-0.3)
\psframe[linewidth=0.02,fillstyle=solid,fillcolor=lightgray](1.5,-0.6)(1.8,-0.3)
\psframe[linewidth=0.02,fillstyle=solid,fillcolor=lightgray](1.8,-0.6)(2.1,-0.3)

\psframe[linewidth=0.02](2.1,-0.6)(2.4,-0.3)
\psframe[linewidth=0.02,fillstyle=solid,fillcolor=darkgray](2.4,-0.6)(2.7,-0.3)
\psframe[linewidth=0.02](2.7,-0.6)(3,-0.3)
\psframe[linewidth=0.02,fillstyle=solid,fillcolor=lightgray](3,-0.6)(3.3,-0.3)

\pscircle[fillstyle=solid,fillcolor=black,linewidth=0.0875,linecolor=white](3.8,-0.45){0.14}
\pscircle[fillstyle=solid,fillcolor=black,linewidth=0.0875,linecolor=white](4.3,-0.45){0.14}
\pscircle[fillstyle=solid,fillcolor=black,linewidth=0.0875,linecolor=white](4.8,-0.45){0.14}

\psframe[linewidth=0.02,fillstyle=solid,fillcolor=lightgray](5.3,-0.6)(5.6,-0.3)
\psframe[linewidth=0.02,fillstyle=solid,fillcolor=lightgray](5.6,-0.6)(5.9,-0.3)
\psframe[linewidth=0.02](5.9,-0.6)(6.2,-0.3)
\psframe[linewidth=0.02,fillstyle=solid,fillcolor=darkgray](6.2,-0.6)(6.5,-0.3)

\psframe[linewidth=0.02](6.5,-0.6)(6.8,-0.3)
\psframe[linewidth=0.02,fillstyle=solid,fillcolor=lightgray](6.8,-0.6)(7.1,-0.3)
\psframe[linewidth=0.02,fillstyle=solid,fillcolor=lightgray](7.1,-0.6)(7.4,-0.3)
\psframe[linewidth=0.02,fillstyle=solid,fillcolor=lightgray](7.4,-0.6)(7.7,-0.3)
\psframe[linewidth=0.02,fillstyle=solid,fillcolor=lightgray](7.7,-0.6)(8,-0.3)

\psframe[linewidth=0.02](8,-0.6)(8.3,-0.3)
\psframe[linewidth=0.02,fillstyle=solid,fillcolor=darkgray](8.3,-0.6)(8.6,-0.3)
\psframe[linewidth=0.02](8.6,-0.6)(8.9,-0.3)
\psframe[linewidth=0.02](8.9,-0.6)(9.2,-0.3)
\psframe[linewidth=0.02](9.2,-0.6)(9.5,-0.3)
\psframe[linewidth=0.02](9.5,-0.6)(9.8,-0.3)
\psframe[linewidth=0.02](9.8,-0.6)(10.1,-0.3)

\psframe[linewidth=0.02](10.1,-0.6)(10.4,-0.3)
\psframe[linewidth=0.02,fillstyle=solid,fillcolor=darkgray](10.7,-0.6)(10.4,-0.3)
\psframe[linewidth=0.02](11,-0.6)(10.7,-0.3)

}

{
\rput(-4,-1.95){\scriptsize $s=2,$}

\rput(-2.7,-1.95){\scriptsize $t=1$}

\psframe[linewidth=0.02](-2.1,-2.1)(-1.8,-1.8)
\psframe[linewidth=0.02](-1.5,-2.1)(-1.8,-1.8)
\psframe[linewidth=0.02,fillstyle=solid,fillcolor=darkgray](-1.5,-2.1)(-1.2,-1.8) 
\psframe[linewidth=0.02](-1.2,-2.1)(-0.9,-1.8)
\psframe[linewidth=0.02](-0.9,-2.1)(-0.6,-1.8) 
\psframe[linewidth=0.02](-0.6,-2.1)(-0.3,-1.8)
\psframe[linewidth=0.02](-0.3,-2.1)(0,-1.8) 

\psframe[linewidth=0.02](0,-2.1)(0.3,-1.8)
\psframe[linewidth=0.02](0.3,-2.1)(0.6,-1.8)
\psframe[linewidth=0.02,fillstyle=solid,fillcolor=darkgray](0.6,-2.1)(0.9,-1.8)
\psframe[linewidth=0.02,fillstyle=solid,fillcolor=lightgray](0.9,-2.1)(1.2,-1.8)
\psframe[linewidth=0.02,fillstyle=solid,fillcolor=lightgray](1.2,-2.1)(1.5,-1.8)
\psframe[linewidth=0.02,fillstyle=solid,fillcolor=lightgray](1.5,-2.1)(1.8,-1.8)
\psframe[linewidth=0.02,fillstyle=solid,fillcolor=lightgray](1.8,-2.1)(2.1,-1.8)

\psframe[linewidth=0.02](2.1,-2.1)(2.4,-1.8)
\psframe[linewidth=0.02](2.4,-2.1)(2.7,-1.8)
\psframe[linewidth=0.02,fillstyle=solid,fillcolor=darkgray](2.7,-2.1)(3,-1.8)
\psframe[linewidth=0.02,fillstyle=solid,fillcolor=lightgray](3,-2.1)(3.3,-1.8)

\pscircle[fillstyle=solid,fillcolor=black,linewidth=0.0875,linecolor=white](3.8,-1.95){0.14}
\pscircle[fillstyle=solid,fillcolor=black,linewidth=0.0875,linecolor=white](4.3,-1.95){0.14}
\pscircle[fillstyle=solid,fillcolor=black,linewidth=0.0875,linecolor=white](4.8,-1.95){0.14}

\psframe[linewidth=0.02,fillstyle=solid,fillcolor=lightgray](5.3,-2.1)(5.6,-1.8)
\psframe[linewidth=0.02,fillstyle=solid,fillcolor=lightgray](5.6,-2.1)(5.9,-1.8)
\psframe[linewidth=0.02](5.9,-2.1)(6.2,-1.8)
\psframe[linewidth=0.02](6.2,-2.1)(6.5,-1.8)

\psframe[linewidth=0.02,fillstyle=solid,fillcolor=darkgray](6.5,-2.1)(6.8,-1.8)
\psframe[linewidth=0.02,fillstyle=solid,fillcolor=lightgray](6.8,-2.1)(7.1,-1.8)
\psframe[linewidth=0.02,fillstyle=solid,fillcolor=lightgray](7.1,-2.1)(7.4,-1.8)
\psframe[linewidth=0.02,fillstyle=solid,fillcolor=lightgray](7.4,-2.1)(7.7,-1.8)
\psframe[linewidth=0.02,fillstyle=solid,fillcolor=lightgray](7.7,-2.1)(8,-1.8)

\psframe[linewidth=0.02](8,-2.1)(8.3,-1.8)
\psframe[linewidth=0.02](8.3,-2.1)(8.6,-1.8)
\psframe[linewidth=0.02,fillstyle=solid,fillcolor=darkgray](8.6,-2.1)(8.9,-1.8)
\psframe[linewidth=0.02](8.9,-2.1)(9.2,-1.8)
\psframe[linewidth=0.02](9.2,-2.1)(9.5,-1.8)
\psframe[linewidth=0.02](9.5,-2.1)(9.8,-1.8)
\psframe[linewidth=0.02](9.8,-2.1)(10.1,-1.8)

\psframe[linewidth=0.02](10.1,-2.1)(10.4,-1.8)
\psframe[linewidth=0.02](10.7,-2.1)(10.4,-1.8)
\psframe[linewidth=0.02,fillstyle=solid,fillcolor=darkgray](11,-2.1)(10.7,-1.8)

}

\end{pspicture}
\end{center}
 \caption{Structure of joint-transmission scheme, $\ntone = 2$, $\nttwo = 1$, $L = 7$ and $T=2$.}
 \label{fig:pilot_data_illustration_JTD}
 \end{figure*}

Both users transmit codewords and pilot symbols over the channel \eqref{eq:channel}. Codewords are used to convey the messages, and pilot symbols are used to facilitate the estimation of the fading coefficients at the receiver. To transmit the message $m_s \in \{1,\dotsc,e^{nR_s} \}$, $s=1,2$, each user's encoder selects a codeword of length $n$ from a codebook $\mathcal{C}_s$, where $\mathcal{C}_s$, $s=1,2$ are drawn i.i.d.\ from an $\nts$-variate, zero-mean, complex Gaussian distribution of covariance matrix ${\sf I}_{\nts}$.

Orthogonal pilot vectors are used to estimate the fading matrices for both users. The pilot vector ${\bm p}_{s,t} \in \field^{\nts}$, $s=1,2$, $t=1,\dotsc,\nts$ used to estimate the fading coefficients from transmit antenna $t$ of user $s$ is given by $p_{s,t} (t) = 1$ and $p_{s,t} (t') = 0$ for $t' \neq t$. For example, the first pilot vector of user $s$ is given by $\trans{(1,0,\dotsc,0)}$, where $\trans{(\cdot)}$ denotes the transpose. To estimate the fading matrices $\HRM_{1,k}$ and $\HRM_{2,k}$, each training period requires the $(\ntone + \nttwo)$ pilot vectors ${\bm p}_{1,1},\dotsc, {\bm p}_{1,\ntone}, {\bm p}_{2,1},\dotsc, {\bm p}_{2,\nttwo}$.

The transmission scheme for the two-user setup extends the scheme used for the single-user setup in \cite{IEEE:asyhari_etal:nearest_neighbour_ISIT11}. We assume that the transmission from both users is synchronised. Every $L$ time instants (for some $L \geq \ntone + \nttwo,~ L\in\integ$), user 1 first transmits the $\ntone$ pilot vectors ${\bm p}_{1,1},\dotsc, {\bm p}_{1,\ntone}$. Once the transmission of the $\ntone$ pilot vectors is finished, the user 2 transmits its $\nttwo$ pilot vectors ${\bm p}_{2,1},\dotsc, {\bm p}_{2,\nttwo}$.  The codeword for each user is then split up into blocks of $(L- \ntone -\nttwo)$ data vectors, which will be transmitted after the $(\ntone + \nttwo)$ pilot vectors. The process of transmitting $(L- \ntone -\nttwo)$ data vectors and $(\ntone + \nttwo)$ pilot vectors continues until all $n$ data symbols are completed. Herein we assume that $n$ is an integer multiple of $(L - \ntone - \nttwo)$.\footnote{If $n$ is not an integer multiple of $(L-\ntone-\nttwo)$, then the last $(L-\ntone-\nttwo)$ instants are not fully used by data vectors and contain therefore time instants where we do not transmit anything. The thereby incurred loss in information rate vanishes as $n$ tends to infinity.} Prior to transmitting the first data block, and after transmitting the last data block, we introduce a guard period of $L (T -1)$ time instants (for some $T\in\integ$), where we transmit every $L$ time instants the $(\ntone + \nttwo)$ pilot vectors but we do not transmit data vectors in between. The guard period ensures that, at every time instant, we can employ a channel estimator that bases its estimation on the channel outputs corresponding to the $T$ past and the $T$ future pilot transmissions. This facilitates the analysis but does not incur a loss in performance. The above transmission scheme is illustrated in Fig. \ref{fig:pilot_data_illustration_JTD}. The channel estimator is described below. 

Note that the total block-length of the above transmission scheme (comprising data vectors, pilot vectors and guard period) is given by
\begin{equation}
 n' = n_{\rm p} + n + n_{\rm g} \label{eq:total_length}
\end{equation}
where $n_{\rm p}$ denotes the number of channel uses for pilot symbols, and where $n_{\rm g}$ denotes the number of channel uses during the guard period, i.e.,
\begin{eqnarray}
 n_{\rm p} &=& \left(\frac{n}{L - \ntone - \nttwo} + 1  + 2 (T-1) \right) (\ntone + \nttwo) \IEEEeqnarraynumspace\\
 n_{\rm g} &=& 2(L- \ntone - \nttwo)(T-1).
\end{eqnarray}
 
Once the transmission is completed, the decoder guesses which message has been transmitted. The decoder consists of two parts: a \emph{channel estimator} and a \emph{data detector}. The channel estimator observes the channel output ${\yrv}_k$, $k\in\mathcal{P}$ corresponding to the past and future $T$ pilot transmissions and estimates $H_{s,k}(r,t)$ using a linear interpolator, i.e., the estimate $\hat{H}_{s,k}^{(T)}(r,t)$ of the fading coefficient $H_{s,k}(r,t)$ is given by
\begin{equation}
\label{eq:LMMSEestimation}
 \hat H_{s,k}^{(T)}(r,t) = \sum^{k + T L}_{\substack{ k' = k - TL:\\ k' \in \mathcal{P}}} a_{s,k'}  (r,t) Y_{k'}(r)
\end{equation}
where the coefficients $a_{s,k'} (r,t)$ are chosen in order to minimise the mean-squared error. Here $\mathcal{P}$ denotes the set of time indices where pilot symbols are transmitted, and $\mathcal{D}$ denotes the set of time indices where data vectors of a codeword are transmitted.

Note that, since the pilot symbols are transmitted only from one user and one antenna at a time, the fading coefficients corresponding to all transmit and receive antennas from both users can be observed. Further note that, since the fading processes $\{H_{s,k}(r,t), k \in \integ \}$, $s=1,2$, $r=1,\ldots,\nr$, $t=1,\ldots,\nts$ are independent, estimating $H_{s,k}(r,t)$ only based on $\{Y_{k}(r),k\in\integ\}$ rather than on $\{\yrv_{k},k\in\integ\}$ incurs no loss in optimality.

 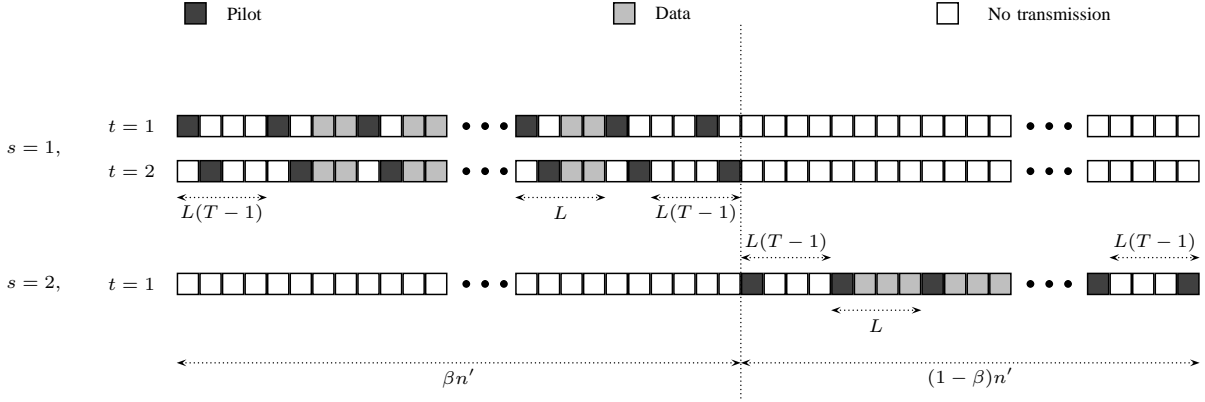
\begin{figure*}[t]
 \begin{center}
\begin{pspicture}(-5,-3.65)(12,2)

\psframe[linewidth=0.02,fillstyle=solid,fillcolor=darkgray](-2,1.5)(-1.7,1.8)
\rput(-1.2,1.65){\scriptsize Pilot}

\psframe[linewidth=0.02,fillstyle=solid,fillcolor=lightgray](3.7,1.5)(4,1.8)
\rput(4.5,1.65){\scriptsize Data}

\psframe[linewidth=0.02](8.0,1.5)(8.3,1.8)
\rput(9.5,1.65){\scriptsize No transmission}


\rput(-1.5,-1){\scriptsize $L(T-1)$}
\psline[linewidth=0.02,linestyle=dotted,dotsep=1pt]{<->}(-2.1,-0.8)(-0.9,-0.8)

\rput(3,-1){\scriptsize $L$}
\psline[linewidth=0.02,linestyle=dotted,dotsep=1pt]{<->}(2.4,-0.8)(3.6,-0.8)

\rput(4.8,-1){\scriptsize $L(T-1)$}
\psline[linewidth=0.02,linestyle=dotted,dotsep=1pt]{<->}(4.2,-0.8)(5.4,-0.8)

\rput(6,-1.4){\scriptsize $L(T-1)$}
\psline[linewidth=0.02,linestyle=dotted,dotsep=1pt]{<->}(5.4,-1.6)(6.6,-1.6)

\rput(7.2,-2.5){\scriptsize $L$}
\psline[linewidth=0.02,linestyle=dotted,dotsep=1pt]{<->}(6.6,-2.3)(7.8,-2.3)

\rput(10.9,-1.4){\scriptsize $L(T-1)$}
\psline[linewidth=0.02,linestyle=dotted,dotsep=1pt]{<->}(10.3,-1.6)(11.5,-1.6)

\psline[linewidth=0.02,linestyle=dotted,dotsep=1pt](5.4,1.5)(5.4,-3.5)

\rput(1.65,-3.2){\scriptsize $\beta n'$} 
\psline[linewidth=0.02,linestyle=dotted,dotsep=1pt]{<->}(-2.1,-3)(5.4,-3)

\rput(8.45,-3.2){\scriptsize $(1 - \beta) n'$} 
\psline[linewidth=0.02,linestyle=dotted,dotsep=1pt]{<->}(5.4,-3)(11.5,-3)

\rput(-4,-0.15){\scriptsize $s=1,$}

{
\rput(-2.7,0.15){\scriptsize $t=1$}

\psframe[linewidth=0.02,fillstyle=solid,fillcolor=darkgray](-2.1,0)(-1.8,0.3)
\psframe[linewidth=0.02](-1.8,0)(-1.5,0.3)
\psframe[linewidth=0.02](-1.5,0)(-1.2,0.3)
\psframe[linewidth=0.02](-1.2,0)(-0.9,0.3)
\psframe[linewidth=0.02,fillstyle=solid,fillcolor=darkgray](-0.9,0)(-0.6,0.3)
\psframe[linewidth=0.02](-0.6,0)(-0.3,0.3)
\psframe[linewidth=0.02,fillstyle=solid,fillcolor=lightgray](-0.3,0)(0,0.3)
\psframe[linewidth=0.02,fillstyle=solid,fillcolor=lightgray](0,0)(0.3,0.3)
\psframe[linewidth=0.02,fillstyle=solid,fillcolor=darkgray](0.3,0)(0.6,0.3)
\psframe[linewidth=0.02](0.6,0)(0.9,0.3)
\psframe[linewidth=0.02,fillstyle=solid,fillcolor=lightgray](0.9,0)(1.2,0.3)
\psframe[linewidth=0.02,fillstyle=solid,fillcolor=lightgray](1.2,0)(1.5,0.3)

\pscircle[fillstyle=solid,fillcolor=black,linewidth=0.0875,linecolor=white](1.75,0.15){0.14}
\pscircle[fillstyle=solid,fillcolor=black,linewidth=0.0875,linecolor=white](2,0.15){0.14}
\pscircle[fillstyle=solid,fillcolor=black,linewidth=0.0875,linecolor=white](2.25,0.15){0.14}

\psframe[linewidth=0.02,fillstyle=solid,fillcolor=darkgray](2.4,0)(2.7,0.3) 
\psframe[linewidth=0.02](2.7,0)(3,0.3)
\psframe[linewidth=0.02,fillstyle=solid,fillcolor=lightgray](3,0)(3.3,0.3)
\psframe[linewidth=0.02,fillstyle=solid,fillcolor=lightgray](3.3,0)(3.6,0.3) 
\psframe[linewidth=0.02,fillstyle=solid,fillcolor=darkgray](3.6,0)(3.9,0.3)
\psframe[linewidth=0.02](3.9,0)(4.2,0.3)
\psframe[linewidth=0.02](4.2,0)(4.5,0.3)
\psframe[linewidth=0.02](4.5,0)(4.8,0.3)
\psframe[linewidth=0.02,fillstyle=solid,fillcolor=darkgray](4.8,0)(5.1,0.3)
\psframe[linewidth=0.02](5.1,0)(5.4,0.3)
\psframe[linewidth=0.02](5.4,0)(5.7,0.3)
\psframe[linewidth=0.02](5.7,0)(6,0.3)
\psframe[linewidth=0.02](6,0)(6.3,0.3)
\psframe[linewidth=0.02](6.3,0)(6.6,0.3)
\psframe[linewidth=0.02](6.6,0)(6.9,0.3)
\psframe[linewidth=0.02](6.9,0)(7.2,0.3)
\psframe[linewidth=0.02](7.2,0)(7.5,0.3)
\psframe[linewidth=0.02](7.5,0)(7.8,0.3)
\psframe[linewidth=0.02](7.8,0)(8.1,0.3)
\psframe[linewidth=0.02](8.1,0)(8.4,0.3)
\psframe[linewidth=0.02](8.4,0)(8.7,0.3)
\psframe[linewidth=0.02](8.7,0)(9,0.3)

\pscircle[fillstyle=solid,fillcolor=black,linewidth=0.0875,linecolor=white](9.25,0.15){0.14}
\pscircle[fillstyle=solid,fillcolor=black,linewidth=0.0875,linecolor=white](9.5,0.15){0.14}
\pscircle[fillstyle=solid,fillcolor=black,linewidth=0.0875,linecolor=white](9.75,0.15){0.14}

\psframe[linewidth=0.02](10.0,0)(10.3,0.3)
\psframe[linewidth=0.02](10.3,0)(10.6,0.3)
\psframe[linewidth=0.02](10.6,0)(10.9,0.3)
\psframe[linewidth=0.02](10.9,0)(11.2,0.3)
\psframe[linewidth=0.02](11.2,0)(11.5,0.3)

}

{
\rput(-2.7,-0.45){\scriptsize $t=2$}

\psframe[linewidth=0.02](-2.1,-0.6)(-1.8,-0.3)
\psframe[linewidth=0.02,fillstyle=solid,fillcolor=darkgray](-1.5,-0.6)(-1.8,-0.3)
\psframe[linewidth=0.02](-1.5,-0.6)(-1.2,-0.3) 
\psframe[linewidth=0.02](-1.2,-0.6)(-0.9,-0.3)
\psframe[linewidth=0.02](-0.9,-0.6)(-0.6,-0.3) 
\psframe[linewidth=0.02,fillstyle=solid,fillcolor=darkgray](-0.6,-0.6)(-0.3,-0.3)
\psframe[linewidth=0.02,fillstyle=solid,fillcolor=lightgray](-0.3,-0.6)(0,-0.3) 
\psframe[linewidth=0.02,fillstyle=solid,fillcolor=lightgray](0,-0.6)(0.3,-0.3)
\psframe[linewidth=0.02](0.3,-0.6)(0.6,-0.3)
\psframe[linewidth=0.02,fillstyle=solid,fillcolor=darkgray](0.6,-0.6)(0.9,-0.3)
\psframe[linewidth=0.02,fillstyle=solid,fillcolor=lightgray](0.9,-0.6)(1.2,-0.3)
\psframe[linewidth=0.02,fillstyle=solid,fillcolor=lightgray](1.2,-0.6)(1.5,-0.3)

\pscircle[fillstyle=solid,fillcolor=black,linewidth=0.0875,linecolor=white](1.75,-0.45){0.14}
\pscircle[fillstyle=solid,fillcolor=black,linewidth=0.0875,linecolor=white](2,-0.45){0.14}
\pscircle[fillstyle=solid,fillcolor=black,linewidth=0.0875,linecolor=white](2.25,-0.45){0.14}

\psframe[linewidth=0.02](2.4,-0.6)(2.7,-0.3)
\psframe[linewidth=0.02,fillstyle=solid,fillcolor=darkgray](2.7,-0.6)(3,-0.3)
\psframe[linewidth=0.02,fillstyle=solid,fillcolor=lightgray](3,-0.6)(3.3,-0.3)
\psframe[linewidth=0.02,fillstyle=solid,fillcolor=lightgray](3.3,-0.6)(3.6,-0.3)
\psframe[linewidth=0.02](3.6,-0.6)(3.9,-0.3)
\psframe[linewidth=0.02,fillstyle=solid,fillcolor=darkgray](3.9,-0.6)(4.2,-0.3)
\psframe[linewidth=0.02](4.2,-0.6)(4.5,-0.3)
\psframe[linewidth=0.02](4.5,-0.6)(4.8,-0.3)
\psframe[linewidth=0.02](4.8,-0.6)(5.1,-0.3)
\psframe[linewidth=0.02,fillstyle=solid,fillcolor=darkgray](5.1,-0.6)(5.4,-0.3)
\psframe[linewidth=0.02](5.4,-0.6)(5.7,-0.3)
\psframe[linewidth=0.02](5.7,-0.6)(6,-0.3)
\psframe[linewidth=0.02](6,-0.6)(6.3,-0.3)
\psframe[linewidth=0.02](6.3,-0.6)(6.6,-0.3)
\psframe[linewidth=0.02](6.6,-0.6)(6.9,-0.3)
\psframe[linewidth=0.02](6.9,-0.6)(7.2,-0.3)
\psframe[linewidth=0.02](7.2,-0.6)(7.5,-0.3)
\psframe[linewidth=0.02](7.5,-0.6)(7.8,-0.3)
\psframe[linewidth=0.02](7.8,-0.6)(8.1,-0.3)
\psframe[linewidth=0.02](8.1,-0.6)(8.4,-0.3)
\psframe[linewidth=0.02](8.4,-0.6)(8.7,-0.3)
\psframe[linewidth=0.02](8.7,-0.6)(9,-0.3)

\pscircle[fillstyle=solid,fillcolor=black,linewidth=0.0875,linecolor=white](9.25,-0.45){0.14}
\pscircle[fillstyle=solid,fillcolor=black,linewidth=0.0875,linecolor=white](9.5,-0.45){0.14}
\pscircle[fillstyle=solid,fillcolor=black,linewidth=0.0875,linecolor=white](9.75,-0.45){0.14}

\psframe[linewidth=0.02](10,-0.6)(10.3,-0.3)
\psframe[linewidth=0.02](10.3,-0.6)(10.6,-0.3)
\psframe[linewidth=0.02](10.6,-0.6)(10.9,-0.3)
\psframe[linewidth=0.02](10.9,-0.6)(11.2,-0.3)
\psframe[linewidth=0.02](11.2,-0.6)(11.5,-0.3)

}

{
\rput(-4,-1.95){\scriptsize $s=2,$}

\rput(-2.7,-1.95){\scriptsize $t=1$}

\psframe[linewidth=0.02](-2.1,-2.1)(-1.8,-1.8)
\psframe[linewidth=0.02](-1.5,-2.1)(-1.8,-1.8)
\psframe[linewidth=0.02](-1.5,-2.1)(-1.2,-1.8) 
\psframe[linewidth=0.02](-1.2,-2.1)(-0.9,-1.8)
\psframe[linewidth=0.02](-0.9,-2.1)(-0.6,-1.8) 
\psframe[linewidth=0.02](-0.6,-2.1)(-0.3,-1.8)
\psframe[linewidth=0.02](-0.3,-2.1)(0,-1.8) 

\psframe[linewidth=0.02](0,-2.1)(0.3,-1.8)
\psframe[linewidth=0.02](0.3,-2.1)(0.6,-1.8)
\psframe[linewidth=0.02](0.6,-2.1)(0.9,-1.8)
\psframe[linewidth=0.02](0.9,-2.1)(1.2,-1.8)
\psframe[linewidth=0.02](1.2,-2.1)(1.5,-1.8)

\pscircle[fillstyle=solid,fillcolor=black,linewidth=0.0875,linecolor=white](1.75,-1.95){0.14}
\pscircle[fillstyle=solid,fillcolor=black,linewidth=0.0875,linecolor=white](2,-1.95){0.14}
\pscircle[fillstyle=solid,fillcolor=black,linewidth=0.0875,linecolor=white](2.25,-1.95){0.14}

\psframe[linewidth=0.02](2.4,-2.1)(2.7,-1.8)
\psframe[linewidth=0.02](2.7,-2.1)(3,-1.8)
\psframe[linewidth=0.02](3,-2.1)(3.3,-1.8)

\psframe[linewidth=0.02](3.3,-2.1)(3.6,-1.8)
\psframe[linewidth=0.02](3.6,-2.1)(3.9,-1.8)
\psframe[linewidth=0.02](3.9,-2.1)(4.2,-1.8)
\psframe[linewidth=0.02](4.2,-2.1)(4.5,-1.8)
\psframe[linewidth=0.02](4.5,-2.1)(4.8,-1.8)
\psframe[linewidth=0.02](4.8,-2.1)(5.1,-1.8)
\psframe[linewidth=0.02](5.1,-2.1)(5.4,-1.8)
\psframe[linewidth=0.02,fillstyle=solid,fillcolor=darkgray](5.4,-2.1)(5.7,-1.8)
\psframe[linewidth=0.02](5.7,-2.1)(6,-1.8)

\psframe[linewidth=0.02](6,-2.1)(6.3,-1.8)
\psframe[linewidth=0.02](6.3,-2.1)(6.6,-1.8)
\psframe[linewidth=0.02,fillstyle=solid,fillcolor=darkgray](6.6,-2.1)(6.9,-1.8)
\psframe[linewidth=0.02,fillstyle=solid,fillcolor=lightgray](6.9,-2.1)(7.2,-1.8)
\psframe[linewidth=0.02,fillstyle=solid,fillcolor=lightgray](7.2,-2.1)(7.5,-1.8)
\psframe[linewidth=0.02,fillstyle=solid,fillcolor=lightgray](7.5,-2.1)(7.8,-1.8)
\psframe[linewidth=0.02,fillstyle=solid,fillcolor=darkgray](7.8,-2.1)(8.1,-1.8)
\psframe[linewidth=0.02,fillstyle=solid,fillcolor=lightgray](8.1,-2.1)(8.4,-1.8)
\psframe[linewidth=0.02,fillstyle=solid,fillcolor=lightgray](8.4,-2.1)(8.7,-1.8)
\psframe[linewidth=0.02,fillstyle=solid,fillcolor=lightgray](8.7,-2.1)(9,-1.8)

\pscircle[fillstyle=solid,fillcolor=black,linewidth=0.0875,linecolor=white](9.25,-1.95){0.14}
\pscircle[fillstyle=solid,fillcolor=black,linewidth=0.0875,linecolor=white](9.5,-1.95){0.14}
\pscircle[fillstyle=solid,fillcolor=black,linewidth=0.0875,linecolor=white](9.75,-1.95){0.14}

\psframe[linewidth=0.02,fillstyle=solid,fillcolor=darkgray](10,-2.1)(10.3,-1.8)
\psframe[linewidth=0.02](10.3,-2.1)(10.6,-1.8)
\psframe[linewidth=0.02](10.6,-2.1)(10.9,-1.8)
\psframe[linewidth=0.02](10.9,-2.1)(11.2,-1.8)
\psframe[linewidth=0.02,fillstyle=solid,fillcolor=darkgray](11.2,-2.1)(11.5,-1.8)

}

\end{pspicture}
\end{center}
 \caption{Structure of TDMA scheme, $\ntone = 2$, $\nttwo = 1$, $L = 4$ and $T=2$.}
 \label{fig:pilot_data_illustration_TDMA}
 \end{figure*}

Since the time-lags between $\HRM_{s,k}$, $k\in\mathcal{D}$ and the observations $\yrv_{k'}$, $k'\in\mathcal{P}$ depend on $k$, it follows that the interpolation error
\begin{equation}
E_{s,k}^{(T)}(r,t) = H_{s,k}(r,t) - \hat H_{s,k}^{(T)}(r,t) \label{eq:Ert}
\end{equation}
is not stationary but cyclo-stationary with period $L$. Nevertheless, it can be shown that, irrespective of $s$, $r$ and $t$, the interpolation error
\begin{equation}
\epsilon^2_{s,T}(\ell,r,t) = \mathsf{E}\left[\left|H_{s,k}(r,t)-\hat{H}^{(T)}_{s,k}(r,t)\right|^2\right]
\end{equation}
tends to the following expression as $T$ tends to infinity \cite{IEEE:ohno:averageratePSAM}
\begin{eqnarray}
 \epsilon^2 (\ell) & \triangleq & \lim_{T \rightarrow \infty} \epsilon^2_{s,T}(\ell,r,t) \\ 
 &=& 1 -  \int^{1/2}_{-1/2} \frac{\SNR \left | f_{H_L,\ell}(\lambda) \right |^2}{\SNR f_{H_L,0} (\lambda) + 1} d \lambda \label{eq:interpolationerror}
\end{eqnarray}
where $\ell= k\mod L$ denotes the remainder of $k/L$. Here $f_{H_L,\ell}(\cdot)$ is given by
\begin{equation}
 f_{H_{L},\ell} (\lambda) = \frac{1}{L} \sum^{L-1}_{j=0} \bar{f}_H \left( \frac{\lambda - j}{L} \right)e^{i2\pi \ell \frac{\lambda - j}{L} }
\end{equation}
 and $\bar{f}_H(\cdot)$ is the periodic function of period $[-1/2,1/2)$ that coincides with $f_H(\lambda)$ for $-1/2\leq\lambda\leq 1/2$. Furthermore, if
\begin{equation}
 L \leq \frac{1}{2 \lambda_D}  \label{eq:nyquist}
\end{equation}
then $|f_{H_L,\ell}(\cdot)|$ becomes
\begin{equation}
|f_{H_L,\ell}(\lambda)| = f_{H_L,0}(\lambda)=\frac{1}{L} f_H\left(\frac{\lambda}{L}\right), \quad\!\! -\frac{1}{2}\leq\lambda\leq\frac{1}{2}.
\end{equation}
In this case the interpolation error \eqref{eq:interpolationerror} becomes
\begin{IEEEeqnarray}{rCll}
 \epsilon^2  &=& 1 -  \int^{1/2}_{-1/2} &\frac{\SNR \left(f_{H}(\lambda) \right)^2}{\SNR f_{H} (\lambda) + L} d \lambda\label{eq:fadingestimate_error}
\end{IEEEeqnarray}
which does not depend on $\ell$ and vanishes as the $\SNR$ tends to infinity. Recall that $\lambda_D$ denotes the bandwidth of $f_H(\cdot)$. Thus, \eqref{eq:nyquist} implies that no aliasing occurs as we undersample the fading process $L$ times.

From the received codeword $\{\yv_k, k \in \integ \}$ and the channel-estimate matrices $\{\hat \HM_{s,k}^{(T)}, k \in \mathcal{D} \}$, $s=1,2$ (which are composed of the entries $\hat h_{s,k}^{(T)}(r,t)$), the decoder chooses the pair of messages $(\hat m_1, \hat m_2)$ that minimises the following distance metric
\begin{IEEEeqnarray}{rCl}
(\hat m_1, \hat m_2) & = & \arg \min_{(m_1, m_2)} D(m_1,m_2) 
\end{IEEEeqnarray}
where 
\begin{IEEEeqnarray}{rCl}
D(m_1,m_2) & \triangleq & \sum_{k \in \mathcal{D} } \biggl\|\yv_k - \sqrt{\SNR}\, \hat \HM^{(T)}_{1,k} \xv_{1,k} (m_1)\nonumber\\
& & \qquad\qquad {} - \sqrt{\SNR}\, \hat \HM^{(T)}_{2,k}\xv_{2,k} (m_2) \biggr\|^2.
 \label{eq:decoding-rule-JTD}
\end{IEEEeqnarray}
In the following, we will refer to the above communication scheme as the \emph{joint-transmission scheme}.

We shall compare the joint-transmission scheme with a time-division multiple access (TDMA) scheme, where each user transmits its message using the transmission scheme shown in Fig.~\ref{fig:pilot_data_illustration_TDMA}. In particular, during the first $\beta n'$ channel uses, for some $0\leq \beta \leq 1$, user 1 transmits its codeword according to the transmission scheme given in \cite{IEEE:asyhari_etal:nearest_neighbour_ISIT11} (see also Fig.~\ref{fig:pilot_data_illustration_TDMA}), while user 2 is silent. (Here $n'$ is given in \eqref{eq:total_length}.) Then, during the next $(1 - \beta)n'$ channel uses, user 2 transmits its codeword according to the same transmission scheme, while user 1 is silent. In both cases, the receiver guesses the corresponding message $m_s$, $s=1,2$ using a nearest neighbour decoder and pilot-assisted channel estimation.

\section{The MAC Pre-Log}
\label{sec:prelog}

Let $R_1^*(\SNR)$, $R_2^*(\SNR)$ and $R_{1 + 2}^* (\SNR)$ be the maximum achievable rate for user 1, the maximum achievable rate for user 2 and the maximum sum rate, respectively. The achievable-rate region is given by the closure of the convex hull of the set \cite{cover_elements_inf_theory}
\begin{IEEEeqnarray}{rCll}
\mathcal{R} & = & \Big \{ &R_1(\SNR),R_2 (\SNR) \colon\nonumber\\
& & & R_1(\SNR) < R_1^* (\SNR),\nonumber\\
& & & R_2 (\SNR) < R_2^* (\SNR),\nonumber\\
& & & R_1 (\SNR) + R_2 (\SNR) < R_{1+2}^* (\SNR) \Big\}.
\end{IEEEeqnarray}
We are interested in the pre-logs of $R_1(\SNR)$ and $R_2(\SNR)$, defined as the limiting ratios of $R_1 (\SNR)$ and $R_2 (\SNR)$ to the logarithm of the SNR as the SNR tends to infinity.  Thus, the pre-log region is given by the closure of the convex hull of the set
\begin{IEEEeqnarray}{rCll}
\Pi_{\mathcal{R}} & = &  \Big\{\Pi_{R_1},\Pi_{R_2}\colon & \Pi_{R_1} < \Pi_{R^*_1},\nonumber\\
& & & \Pi_{R_2} < \Pi_{R^*_2},  \nonumber \\
 &  &  & \Pi_{R_1}+\Pi_{R_2}<\Pi_{R^*_{1+2}}\Bigr\}
\end{IEEEeqnarray}
where
\begin{IEEEeqnarray}{rCl}
&  &  \Pi_{R^*_1}  \triangleq  \limsup_{\SNR \rightarrow \infty} ~ \frac{R^*_1(\SNR)}{\log \SNR}, \label{def:gmi-pre-log-constraint-R-1} \\ 
&  & \Pi_{R^*_2}  \triangleq  \limsup_{\SNR \rightarrow \infty} ~ \frac{R^*_2(\SNR)}{\log \SNR}, \label{def:gmi-pre-log-constraint-R-2} \\ 
&  & \Pi_{R^*_{1+2}} \triangleq \limsup_{\SNR \rightarrow \infty} ~ \frac{R^*_{1+2} (\SNR)}{\log \SNR}. \label{def:gmi-pre-log-constraint-R-1-2} \IEEEeqnarraynumspace
\end{IEEEeqnarray}
The capacity pre-logs $\Pi_{C_1}$ and $\Pi_{C_2}$ are defined in the same way but with $R_1 (\SNR)$ and $R_2 (\SNR)$ replaced by the respective capacities $C_1 (\SNR)$ and $C_2 (\SNR)$.

The pre-log for point-to-point MIMO fading channels has been studied in a number of works, see, e.g., \cite{IEEE:lapidoth:ontheasymptotic-capacity,IEEE:koch:fadingnumber_degreeoffreedom,IEEE:etkin:degreeofffreedomMIMO,IEEE:asyhari_etal:nearest_neighbour_ISIT11}. For example, it was shown in \cite{IEEE:asyhari_etal:nearest_neighbour_ISIT11} that the pre-log of point-to-point $(\nr \times \nt)$-dimensional MIMO fading channels achievable with nearest neighbour decoding and pilot-assisted channel estimation is lower-bounded by 
\begin{equation}
 \Pi_{R^*} \geq \min \left(\nt, \nr \right) \left(1 -  \frac{\min \left(\nt, \nr \right)}{L^*} \right) \label{eq:pre-log-LB-mimo-point-to-point}
\end{equation}
where $L^*$ is the largest integer satisfying $L^* \leq \frac{1}{2\lambda_D}$. It has been observed that, if $1/(2\lambda_D)$ is an integer, then this lower bound coincides with the best so far known lower bound on the capacity pre-log derived by Etkin and Tse \cite{IEEE:etkin:degreeofffreedomMIMO}, namely, 
\begin{IEEEeqnarray}{lCl}
\Pi_{C} & \geq & \min (\nt, \nr) \Big(1 - \min(\nt,\nr) \mu\big(\{ \lambda\colon f_H(\lambda)> 0\}\big) \Big)\label{eq:etaMIMO}\IEEEeqnarraynumspace
\end{IEEEeqnarray}
where $\mu (\cdot)$ denotes the Lebesgue measure on the interval $[-1/2,1/2]$. For MISO fading channels, the lower bound \eqref{eq:pre-log-LB-mimo-point-to-point} specialises to
\begin{equation}
 \Pi_{R_1^*} \geq 1 - \frac{1}{L^*}
\end{equation}
which coincides with the capacity pre-log
\begin{equation}
 \Pi_{C}  =  \mu \big(\{\lambda\colon f_H (\lambda) = 0\} \big) \label{eq:capacity_pre_log_UB}
\end{equation}
derived by Koch and Lapidoth \cite{IEEE:koch:fadingnumber_degreeoffreedom} for MISO channels when $1/(2\lambda_D)$ is an integer. Thus, for point-to-point MISO fading channels, and if $1/(2\lambda_D)$ is an integer, then the communication scheme described in Section~\ref{sec:trans_scheme} achieves the capacity pre-log. 

In the following theorem, we present our result on the pre-log region of the two-user MIMO fading MAC achievable with the joint-transmission scheme.

\begin{theorem}
\label{th:JTD_pre-log}
 Consider the MIMO fading MAC model \eqref{eq:channel}. Then, the pre-log region achievable with the joint-transmission scheme described in Section \ref{sec:trans_scheme} is given by the closure of the convex hull of the set
 \begin{IEEEeqnarray}{rCll}
& &  \Bigg  \{ &   \Pi_{R_1}, \Pi_{R_2}   \colon \nonumber \\
& & & \,\, \Pi_{R_1}  <  \min \left( \nr, \ntone \right) \left(1 - \frac{\ntone + \nttwo}{L^*} \right), \nonumber \\
& & & \,\,\Pi_{R_2}  <  \min \left( \nr, \nttwo \right) \left(1 - \frac{\ntone + \nttwo}{L^*} \right), \nonumber \\
& & & \,\,\Pi_{R_1} + \Pi_{R_2}   <  \min \left( \nr, \ntone + \nttwo \right) \left(1 - \frac{\ntone + \nttwo}{L^*} \right) \Bigg \} \nonumber \\
\label{eq:thJTD_pre-log}
\end{IEEEeqnarray}
where $L^*$ is the largest integer satisfying $L^* \leq \frac{1}{2\lambda_D}$. 
\end{theorem}
\begin{proof}
 An outline of the proof is given in Section \ref{sec:proofs}.
\end{proof}

\begin{remark}
The pre-log region given in Theorem \ref{th:JTD_pre-log} is the largest region achievable with any transmission scheme that uses $(\ntone+\nttwo)/L^*$ of the time for transmitting pilot symbols. Indeed, even if the channel estimator would be able to estimate the fading coefficients perfectly, and even if we could decode the data symbols using a maximum-likelihood decoder, the capacity pre-log region (without pilot transmission) would be given by the closure of the convex hull of the set \cite{Bell_foschini_layered_space-time,telatar_multiantenna_Gaussian,cover_elements_inf_theory}
\begin{IEEEeqnarray}{rCll}
 & & \Big\{(\Pi_{R_1},\Pi_{R_2})\colon &\Pi_{R_1} < \min(\nr, \ntone) \nonumber\\
 & & & \Pi_{R_2} < \min(\nr, \nttwo) \nonumber\\
 & & & \Pi_{R_1}+\Pi_{R_2} < \min(\nr, \ntone + \nttwo) \Big\}\IEEEeqnarraynumspace
\end{IEEEeqnarray}
which, after multiplying by $1-(\ntone+\nttwo)/L^*$ in order to take the pilot symbols into account, becomes \eqref{eq:thJTD_pre-log}. Thus, in order to improve upon \eqref{eq:thJTD_pre-log}, one needs to design a transmission scheme that employs less than $(\ntone+\nttwo)/L^*$ pilot symbols per channel use.
\end{remark}

\begin{remark}[TDMA Pre-Log]
\label{rmrk:TDMA_pre-log}
Consider the MIMO fading MAC model \eqref{eq:channel}. Then, the pre-log region achievable with the TDMA scheme described in Section~\ref{sec:trans_scheme} is the closure of the convex hull of the set
\begin{IEEEeqnarray}{rCll}
 & & \Bigg  \{ &  \Pi_{R_1}, \Pi_{R_2}  \colon \nonumber \\
& &  & \,\,\Pi_{R_1}  <  \beta \min  \left( \nr, \ntone \right) \left(1 - \frac{\ntone}{L^*} \right), \nonumber \\
& & & \,\,\Pi_{R_2}  <  ( 1 - \beta) \min  \left( \nr, \nttwo \right) \left(1 - \frac{\nttwo}{L^*} \right), 0 \leq \beta \leq 1  \Bigg \} \IEEEeqnarraynumspace\nonumber\\
\end{IEEEeqnarray}
where $L^*$ is the largest integer satisfying $L^* \leq \frac{1}{2\lambda_D}$. This follows directly from the pre-log of the point-to-point MIMO fading channel \eqref{eq:pre-log-LB-mimo-point-to-point}.
\end{remark}

Note that the sum of the pre-logs $\Pi_{R_1}+\Pi_{R_2}$ is upper-bounded by the capacity pre-log of the point-to-point MIMO fading channel with $(\ntone + \nttwo)$ transmit antennas and $\nr$ receive antennas, since the point-to-point MIMO channel allows for cooperation between the terminals. While the capacity pre-log of point-to-point MIMO fading channels remains an open problem, the capacity pre-log of point-to-point MISO fading channels is known, cf.~\eqref{eq:capacity_pre_log_UB}. It thus follows from \eqref{eq:capacity_pre_log_UB} that
\begin{equation}
\Pi_{R_1}+\Pi_{R_2} \leq 1-2\lambda_D
\end{equation}
which together with the single-user constraints \cite{IEEE:lapidoth:ontheasymptotic-capacity}
\begin{IEEEeqnarray}{lClCl}
\Pi_{R_1} & \leq & \Pi_{C_1} = 1-2\lambda_D \\
\Pi_{R_2} & \leq & \Pi_{C_2} = 1-2\lambda_D
\end{IEEEeqnarray}
implies that, for $\nr=\ntone=\nttwo=1$, TDMA achieves the capacity pre-log region. The next section provides a more detailed comparison between the joint-transmission scheme and TDMA.

\section{Joint-Transmission vs.\ TDMA}
\label{sec:discussion}

In this section, we discuss how the joint-transmission scheme described in Section~\ref{sec:trans_scheme} performs compared to TDMA. To this end, we compare the sum-rate pre-log $\Pi_{R^*_{1+2}}$ of the joint-transmission scheme (Theorem~\ref{th:JTD_pre-log}) with the sum-rate pre-log of the TDMA scheme described in Section~\ref{sec:trans_scheme} as well as with the sum-rate pre-log of TDMA when the receiver has knowledge of the realisations of the fading processes $\{\HRM_{s,k},\,k\in\integ\}$, $s=1,2$. In the former case, the sum-rate pre-log is given in Remark~\ref{rmrk:TDMA_pre-log}, whereas in the latter case it is
\begin{equation}
\label{eq:TDMAcoh}
\Pi_{R^*_{1+2}} = \beta \min(\nr,\ntone) + (1-\beta) \min(\nr,\nttwo).
\end{equation}
The following corollary presents a sufficient condition on $L^*$ under which the sum-rate pre-log of the joint-transmission scheme is strictly larger than the sum-rate pre-log of the coherent TDMA scheme \eqref{eq:TDMAcoh}, as well as a sufficient condition on $L^*$ under which it is strictly smaller than the sum-rate pre-log of the TDMA scheme given in Remark~\ref{rmrk:TDMA_pre-log}. Since \eqref{eq:TDMAcoh} is an upper-bound on the sum-rate pre-log of any TDMA scheme over the MIMO fading MAC \eqref{eq:channel}, and since the sum-rate pre-log given in Remark~\ref{rmrk:TDMA_pre-log} is a lower bound on the sum-rate pre-log of the best TDMA scheme, it follows that the sufficient conditions presented in Corollary~\ref{cor:JTvsTDMA} hold also for the best TDMA scheme.

\begin{corollary}
\label{cor:JTvsTDMA}
Consider the MIMO fading MAC model \eqref{eq:channel}. The joint-transmission scheme described in Section~\ref{sec:trans_scheme} achieves a larger sum-rate pre-log than any TDMA scheme if
\begin{equation}
\label{eq:corJT}
L^* > \frac{\min(\nr,\ntone+\nttwo)(\ntone+\nttwo)}{\min(\nr,\ntone+\nttwo)-\min\bigl(\nr,\max(\ntone,\nttwo)\bigr)}
\end{equation}
where we define $a/0\triangleq \infty$ for every $a>0$. Conversely, the best TDMA scheme achieves a larger sum-rate pre-log than the joint-transmission scheme if
\begin{IEEEeqnarray}{rCl}
L^* & < & \frac{\min(\nr,\ntone+\nttwo)(\ntone+\nttwo)}{\min(\nr,\ntone+\nttwo)-\min(\nr,\ntone,\nttwo)}  \nonumber \\
&  &  - \frac{\min(\ntone\nr,\ntone^2,\nttwo\nr,\nttwo^2)}{\min(\nr,\ntone+\nttwo)-\min(\nr,\ntone,\nttwo)}.  \label{eq:corTDMA}
\end{IEEEeqnarray}
\end{corollary}

Recall that $L^*$ is inversely proportional to the bandwidth of the power spectral density $f_H(\cdot)$, which in turn is inversely proportional to the coherence time of the fading channel. We thus see from Corollary~\ref{cor:JTvsTDMA} that the joint-transmission scheme tends to be superior to TDMA when the coherence time of the channel is large. In contrast, TDMA is superior to the joint-transmission scheme when the coherence time of the channel is small. 

Intuitively, this can be explained by observing that, compared to TDMA, the joint-transmission scheme uses the multiple antennas at the transmitters and at the receiver more efficiently, but requires more pilot symbols to estimate the fading coefficients. Thus, when the coherence time is large, then the number of pilot symbols required to estimate the fading is small, so the gain in capacity by using the antennas more efficiently dominates the loss incurred by requiring more pilot symbols. Hence, in this case the joint-transmission scheme is superior to TDMA.

We next evaluate \eqref{eq:corJT} and \eqref{eq:corTDMA} for some particular values of $\nr$, $\ntone$, and $\nttwo$.

\subsection{Receiver Employs Less Antennas Than Transmitters}
Suppose that the number of receive antennas is smaller than the number of transmit antennas, i.e., $\nr\leq\min(\ntone,\nttwo)$. Then, the right-hand sides (RHS) of \eqref{eq:corJT} and \eqref{eq:corTDMA} become $\infty$ and every finite $L^*$ satisfies \eqref{eq:corTDMA}. Thus, if the number of receive antennas is smaller than the number of transmit antennas, then, irrespective of $L^*$, TDMA is superior to the joint-transmission scheme.

\subsection{Receiver Employs More Antennas Than Transmitters}

Suppose that the receiver employs more antennas than the transmitters, i.e., $\nr\geq\ntone+\nttwo$, and suppose that $\ntone=\nttwo=\nt$. Then, \eqref{eq:corJT} becomes
\begin{equation}
\label{eq:JT_SIMO}
L^* > 4\nt
\end{equation}
and \eqref{eq:corTDMA} becomes
\begin{equation}
\label{eq:TDMA_SIMO}
L^* < 3\nt.
\end{equation}
Thus, if $L^*$ is greater than $4\nt$, then the joint-transmission scheme is superior to TDMA. In contrast, if $L^*$ is smaller than $3\nt$, then TDMA is superior. Note that if $L^*$ is between $3\nt$ and $4\nt$, then the joint-transmission scheme is superior to the TDMA scheme presented in Section~\ref{sec:trans_scheme}, but it is inferior to the coherent TDMA scheme \eqref{eq:TDMAcoh}. This is illustrated in Fig.~\ref{fig:simo-mac-plot} for the case where $\nr=2$ and $\ntone=\nttwo=1$.

\begin{figure*}[t]
\mbox{
\begin{psfrags}
  \psfrag{a}{\tiny $1- \frac{1}{L^*}$}
 \psfrag{b}{\tiny $1 - \frac{2}{L^*}$}
 \psfrag{c}{\tiny $\Pi_{R_1}$}
 \psfrag{d}{\tiny $\Pi_{R_2}$}
 \psfrag{e}{\tiny $1$}
 \subfloat[$ L^* < 3 $] {\includegraphics[width=0.35\textwidth]{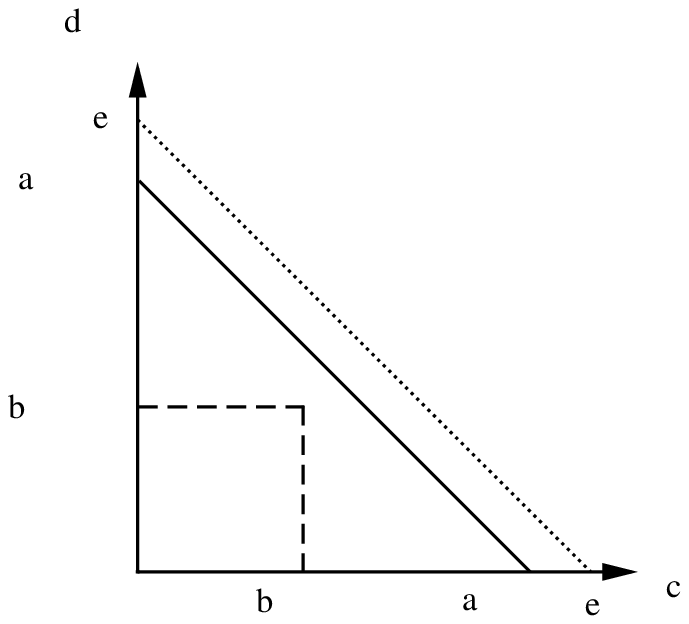}
\label{subfig:simo-mac-alpha-p3-to-p2} }
\end{psfrags}
\begin{psfrags}
 \psfrag{a}{\tiny $1-\frac{1}{L^*}$}
 \psfrag{b}{\tiny $1 - \frac{2}{L^*}$}
 \psfrag{c}{\tiny $\Pi_{R_1}$}
 \psfrag{d}{\tiny $\Pi_{R_2}$}
 \psfrag{e}{\tiny $1$}
  \subfloat[$ L^* > 4$]{ \includegraphics[width=0.35\textwidth]{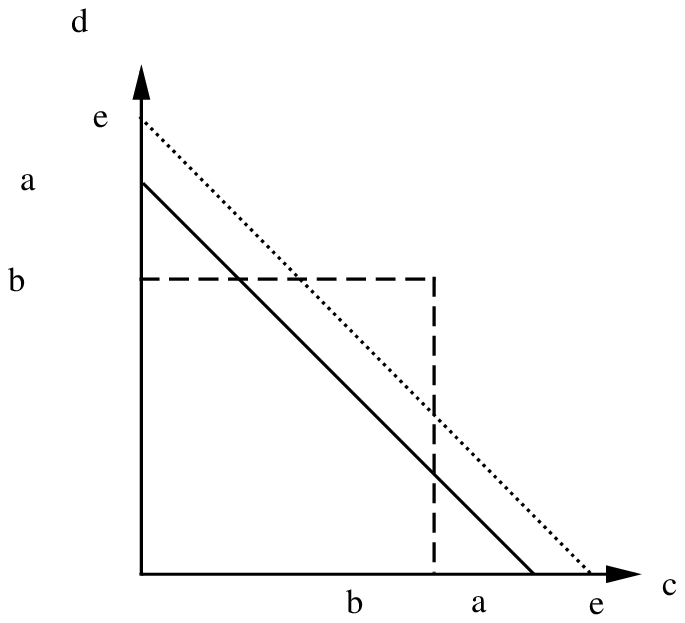}
\label{subfig:simo-mac-alpha-p4}
} 
\end{psfrags}\hspace{-25mm}
 \subfloat{\includegraphics[width=0.4\textwidth]{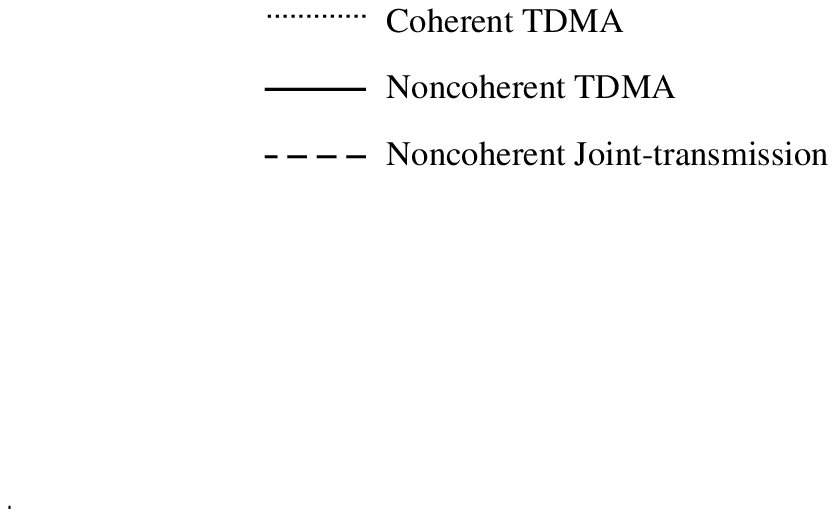}}}
\caption{Pre-log regions for a fading MAC with $\nr=2$ and $\ntone=\nttwo=1$ for different values of $L^*$. Depicted are the pre-log region for the joint-transmission scheme as given in Theorem~\ref{th:JTD_pre-log} (dashed line), the pre-log region of the TDMA scheme as given in Remark~\ref{rmrk:TDMA_pre-log} (solid line), and the pre-log region of the coherent TDMA scheme \eqref{eq:TDMAcoh} (dotted line).}
\label{fig:simo-mac-plot}
\end{figure*}

\subsection{A Case In Between}
Suppose that $\nr\leq\ntone+\nttwo$ and $\nttwo<\nr\leq\ntone$. Then, \eqref{eq:corJT} becomes
\begin{equation}
L^* > \infty
\end{equation}
and \eqref{eq:corTDMA} becomes
\begin{equation}
L^* < \nttwo+\frac{\nr\ntone}{\nr-\nttwo}.
\end{equation}
Thus, in this case the joint-transmission scheme is always inferior to the coherent TDMA scheme \eqref{eq:TDMAcoh}, but it can be superior to the TDMA scheme presented in Section~\ref{sec:trans_scheme}.

\subsection{Typical Values of $L^*$}
We briefly discuss what values of $L^*$ may occur in practical scenarios. To this end, we first recall that $L^*$ is the largest integer satisfying $L^*\leq\frac{1}{2\lambda_D}$, where $\lambda_D$ is the bandwidth of the spectral distribution density $f_H(\cdot)$. It can be associated with the Doppler spread of the channel as
\begin{equation}
\lambda_D = \frac{f_m}{W_c}
\end{equation}
where $f_m$ is the maximum Doppler shift and $W_c$ is the coherence bandwidth of the channel. Following the computations of Etkin and Tse \cite{IEEE:etkin:degreeofffreedomMIMO}, we shall determine typical values of $\lambda_D$ for indoor, urban, and rural environments for carrier frequencies ranging from 800 MHz to 5 GHz. For indoor environments, assuming mobile speeds of 5 km/h, $\lambda_D$ ranges from $2\cdot 10^{-7}$ to $10^{-5}$. For urban environments, assuming the same mobile speeds, $\lambda_D$ ranges from $2\cdot 10^{-5}$ to $2\cdot 10^{-4}$, whereas for mobile speeds of 75 km/h, $\lambda_D$ ranges from $2\cdot 10^{-4}$ to $0.004$. Finally, for rural environments and mobile speeds of 200 km/h, $\lambda_D$ ranges from $0.007$ to $0.05$.

For indoor environments and mobile speeds of 5 km/h, we thus have that $L^*$ is typically greater than $10^5$. For urban environments, $L^*$ is typically greater than $5\cdot 10^3$ for mobile speeds of 5 km/h and greater than $125$ for mobile speeds of 75 km/h. For rural environments and mobile speeds of 200 km/h, $L^*$ ranges typically from $10$ to $71$. Thus, for most practical scenarios, $L^*$ is typically large. It therefore follows that, if $\nr\geq\ntone+\nttwo$, \eqref{eq:corJT} is satisfied unless $\ntone+\nttwo$ is very large. For example, if the receiver employs more antennas than the transmitters, and if $\ntone=\nttwo=\nt$, then $L^* > 4\nt$, is satisfied even for urban environments and mobile speeds of 75 km/h, as long as $\nt<30$. Only for rural environments and mobile speeds of 200 km/h, this condition may not be satisfied for a practical number of transmit antennas. Thus, if the number of antennas at the receiver is sufficiently large, then the joint-transmission scheme is superior to TDMA in most practical scenarios. On the other hand, if $\nr\leq\min(\ntone,\nttwo)$, then TDMA is always superior to the joint-transmission scheme, irrespective of how large $L^*$ is. This suggests that one should use more antennas at the receiver than at the transmitters.

\section{Proof Outline}
\label{sec:proofs}

Since the codebook construction is symmetric, it suffices to study the conditional probability of error, conditioned on the event that the messages $(m_1,m_2)=(1,1)$ were transmitted. Let $\mathcal{E}(m_1', m_2')$ be the event that $D(m_1',m_2') \leq D(1,1)$. The error probability can be upper-bounded by
\begin{equation}
P_e \leq \Pr \left \{ \bigcup_{(m_1', m_2') \neq (1,1)}  \mathcal{E}(m_1', m_2') \right \}.
\end{equation}
This upper bound depends on the three error events $(m_1' \neq 1, m_2' = 1)$, $(m_1' = 1, m_2' \neq 1)$ and $(m_1' \neq 1, m_2' \neq 1)$. 

To prove Theorem \ref{th:JTD_pre-log}, we analyse the generalised mutual information (GMI) for the channel model considered in Section \ref{sec:intro} and the transmission scheme in Section \ref{sec:trans_scheme}. The GMI, denoted by $I^{\rm gmi}(\SNR)$, specifies the highest information rate for which the average probability of error, averaged over the ensemble of i.i.d. Gaussian codebooks, tends to zero as the codeword length $n$ tends to infinity (see \cite{IEEE:lapidoth:nearestneighbournongaussian,IEEE:lapidoth:fadingchannels_howperfect,IEEE:weingarten:gaussiancodes} and references therein). In accordance with the above error events, we consider the following three maximum achievable rates: $\gmi_1 (\SNR)$ and $\gmi_2 (\SNR)$ specify the maximum transmission rate for user 1 ($m'_1 \neq 1, m'_2 = 1$) and user 2 ($m'_1 = 1, m'_2 \neq 1$), respectively, whereas $\gmi_{1+2} (\SNR)$ specifies the maximum sum-rate ($m'_1 \neq 1, m'_2 \neq 1$).

\subsection*{\underline{Error Event $m_1' \neq 1, m_2' = 1$}}

Let $\mathbb{E}_{s,k}^{(T)}$ denote the estimation-error matrix in estimating $\mathbb{H}_{s,k}$, i.e., $\mathbb{E}_{s,k}^{(T)}$ is composed of the entries $E_{s,k}^{(T)}(r,t)$ \eqref{eq:Ert}. Then, the GMI corresponding to the event $\mathcal{E}(m'_1, 1)$, $m'_1 \neq 1$ can be evaluated as \cite{IEEE:lapidoth:fadingchannels_howperfect, IEEE:weingarten:gaussiancodes}
\begin{equation}
I^{\rm gmi}_1 (\SNR) =  \sup_{\theta \leq 0} ~ \bigl(\theta F(\SNR) - \kappa_1 (\theta, \SNR)  \bigr) \label{eq:gmi_user_1_def}
\end{equation}
where 
\begin{IEEEeqnarray}{rCl}
F(\SNR) & = & \frac{\nr(L-\ntone-\nttwo)}{L} \nonumber\\
\IEEEeqnarraymulticol{3}{r}{\qquad {} + \frac{1}{L} \sum^{L-\ntone-\nttwo}_{\ell=1} \mathsf{E} \left[\SNR \left( \left \| \ERM^{(T)}_{1,\ell} \right \|^2_F +  \left \|  \ERM^{(T)}_{2,\ell} \right \|^2_F \right) \right] \IEEEeqnarraynumspace} \label{eq:expectation-F}
\end{IEEEeqnarray}
(with $\| \cdot \|_F$ denoting the Frobenius norm); and where $\kappa_1 (\theta, \SNR)$ is the conditional log moment-generating function of the metric $D (m_1', m_2')$ associated with $m_1' \neq 1$,conditioned on the channel outputs, on $m_2' = 1$ and on the fading estimates, given by
\begin{IEEEeqnarray}{rCl}
\kappa_1 (\theta) & = & \frac{1}{L} \sum^{L- \ntone - \nttwo}_{\ell=1} g_{1,\ell} \nonumber\\
\IEEEeqnarraymulticol{3}{r}{\quad {}  - \frac{1}{L} \sum^{L- \ntone - \nttwo}_{\ell=1} \mathsf{E} \left[ \log {\rm det} \left(\mathsf{I}_\nr - \theta \SNR \hat \HRM_{1,\ell}^{(T)} \hat \HRM_{1,\ell}^{\dagger(T)} \right)\right]\IEEEeqnarraynumspace}
\end{IEEEeqnarray}
where
\begin{IEEEeqnarray}{rCl}
 g_{1,\ell} & \triangleq & \mathsf{E} \bigg[ \theta \left( \yrv_\ell - \sqrt{\SNR} \hat \HRM_{2,\ell}^{(T)} \xrv_{2,\ell} \right)^\dagger\nonumber\\
 & & \qquad {} \times\left( \mathsf{I}_\nr - \theta \SNR \hat \HRM_{1,\ell}^{(T)} \hat \HRM_{1,\ell}^{\dagger(T)}  \right)^{-1}\nonumber\\
 & & \qquad\qquad {} \times \left( \yrv_\ell - \sqrt{\SNR} \hat \HRM_{2,\ell}^{(T)} \xrv_{2,\ell} \right) \bigg].\IEEEeqnarraynumspace
\end{IEEEeqnarray}

Following \cite{IEEE_IT_Asyhari_Nearest_neighbour}, it can be shown that for $\theta \leq 0$ we have $g_{1,\ell} \leq 0$. As observed in \cite{IEEE:asyhari_etal:nearest_neighbour_ISIT11}, the choice
\begin{equation}
 \theta = - \frac{1}{\nr +  \nr \left( \ntone + \nttwo \right)  \SNR \epsilon^2_{T} }
\end{equation}
yields a good lower bound at high SNR. Here
\begin{equation}
\label{eq:estimationerror}
\epsilon^2_{T} = \max_{s,r,t,\ell} ~ \mathsf{E} \left[ \left|E_{s,\ell}^{(T)} (r,t) \right|^2 \right].
\end{equation}
Substituting this choice to the RHS of \eqref{eq:gmi_user_1_def}, and applying $g_{1,\ell} \leq 0$ to upper-bound $\kappa_1 (\theta, \SNR)$, we obtain
\begin{IEEEeqnarray}{rCl}
\gmi_{1} & \geq & \frac{1}{L} \sum^{L-\ntone -\nttwo}_{\ell=1}\mathsf{E} \Biggl[ \log {\rm det} \Biggl(\mathsf{I}_\nr + \nonumber\\
& & \qquad\qquad\qquad\quad\,\, {}  + \frac{\SNR\hat \HRM_{1,\ell}^{(T)} \hat \HRM_{1,\ell}^{\dagger (T)}}{\nr +   \nr \left(\ntone + \nttwo \right)  \SNR \epsilon^2_{T} } \Biggr) \Biggr] \nonumber\\
& & {} -\frac{L-\ntone-\nttwo}{L}. 
 \label{eq:lower_bound_gmi_1_Tfixed} 
\end{IEEEeqnarray}

We continue by analysing the RHS of \eqref{eq:lower_bound_gmi_1_Tfixed} in the limit as the observation window $T$ of the channel estimator tends to infinity. To this end, we note that, for $L\leq\frac{1}{2\lambda_D}$, the interpolation error in \eqref{eq:estimationerror}  tends to \eqref{eq:fadingestimate_error}
\begin{equation}
\lim_{T \rightarrow \infty} \epsilon^2_{T} =  \epsilon^2 = 1 - \int^{1/2}_{-1/2} \frac{\SNR(f_H (\lambda))^2}{\SNR f_H (\lambda) + L} d\lambda.
\end{equation}
It follows that, irrespective of $s$ and $k$, the estimate $\hat \HRM_{s,k}^{(T)}$ tends to $\bar \HRM$ in distribution as $T$ tends to infinity, so
\begin{IEEEeqnarray}{lCl}
\IEEEeqnarraymulticol{3}{l}{\frac{\hat \HRM_{s,k}^{(T)} \hat \HRM_{s,k}^{\dagger (T)}}{ \nr +   \nr \left( \ntone  + \nttwo \right) \SNR \epsilon^2_T}}\nonumber\\
\qquad\qquad\qquad & \stackrel{d}{\rightarrow} & \frac{\bar \HRM \bar \HRM^\dagger}{ \nr + \nr \left( \ntone + \nttwo \right) \SNR \epsilon^2} 
\end{IEEEeqnarray}
where the entries of $\bar \HRM$ are i.i.d., circularly-symmetric, complex Gaussian random variables with zero mean and variance $1 - \epsilon^2$. Since the function ${\sf A} \mapsto \det({\sf I}+{\sf A})$ is continuous and bounded from below, we obtain from Portmanteau's lemma \cite{vdvaart_weakconvergence}  that
\begin{IEEEeqnarray}{rCl}
\IEEEeqnarraymulticol{3}{l}{\lim_{T\to\infty}\gmi_1 (\SNR)} \nonumber \\
&\geq& \frac{L- \ntone  - \nttwo}{L} \Biggl(-1 + \nonumber\\
& & \!\! {} + \mathsf{E} \left[ \log {\rm det} \left(\mathsf{I}_\nr + \frac{\SNR \bar \HRM \bar \HRM^\dagger }{ \nr +   \nr \left( \ntone + \nttwo \right) \SNR \epsilon^2 } \right) \right] \Biggr)  \IEEEeqnarraynumspace\\
&\geq&  \frac{L -  \ntone - \nttwo}{L}\min(\nr,\ntone)\Bigl(\log\SNR \nonumber\\
& & \qquad\qquad\qquad\qquad\,\, {} -  \log\bigl(\nr +  \nr ( \ntone + \nttwo ) \SNR \epsilon^2\bigr) \Bigr) \nonumber\\
& & {} + \frac{L -  \ntone - \nttwo}{L}\Psi \IEEEeqnarraynumspace \label{eq:lower_bound_gmi_1_asymptotic_ntonesmall}
\end{IEEEeqnarray}
where
\begin{equation}
\Psi \triangleq \left\{\begin{array}{ll} \mathsf{E}[\log\det \hermi{\bar\HRM} \bar\HRM]-1, \quad & \nr\geq\ntone \\  \mathsf{E} [\log\det\bar\HRM\hermi{\bar\HRM}]-1, \quad & \nr<\ntone.\end{array}\right.
\end{equation}
Here the last inequality follows by lower-bounding $\log {\rm det} \left( {\sf I} + {\sf A} \right) \geq \log {\rm det} {\sf A}$.

\newcounter{old_equation}
\setcounter{old_equation}{\value{equation}}
\begin{figure*}[t]
\setcounter{equation}{57}
\begin{IEEEeqnarray}{rCl}
\kappa_{1,2} (\theta)
&=& \frac{1}{L} \sum^{L- \ntone - \nttwo}_{\ell =1} \mathsf{E} \left[ \theta \yrv_\ell^\dagger \left( \mathsf{I}_\nr - \theta \SNR \left( \hat \HRM_{1,\ell}^{(T)} \hat \HRM_{1,\ell}^{\dagger(T)} + \hat \HRM_{2,\ell}^{(T)} \hat \HRM_{2,\ell}^{\dagger(T)} \right) \right)^{-1} \yrv_\ell \right]  \IEEEeqnarraynumspace \nonumber \\
&   & {} - \frac{1}{L} \sum^{L- \ntone - \nttwo}_{\ell=1} \mathsf{E} \left[ \log {\rm det} \biggl(\mathsf{I}_\nr - \theta \SNR  \Bigl( \hat \HRM_{1,\ell}^{(T)} \hat \HRM_{1,\ell}^{\dagger(T)} + \hat \HRM_{2,\ell}^{(T)} \hat \HRM_{2,\ell}^{\dagger(T)}  \Bigr) \biggr)   \right]. \label{eq:kappa_1_2}
\end{IEEEeqnarray}
\rule{\textwidth}{0.5pt}
\setcounter{equation}{\value{old_equation}}
\end{figure*}

To compute the pre-log
\begin{equation}
\Pi_{R^*_1} \triangleq \lim_{\SNR \rightarrow \infty} \frac{I^{\rm gmi}_1 (\SNR) }{\log \SNR}
\end{equation}
we first note that, by \cite{eurasip:grant:rayleighfading_multi}, $\Psi$ is finite. We further note that 
\begin{equation}
\SNR\, \epsilon^2  = \int^{1/2}_{-1/2} \frac{\SNR f_H(\lambda) L}{\SNR f_H (\lambda) + L} d\lambda \leq L
\end{equation}
which implies that $\log\big(\nr +  \nr (\ntone + \nttwo ) \SNR\,\epsilon^2   \big)$ is bounded. Thus, computing the ratio of the RHS of \eqref{eq:lower_bound_gmi_1_asymptotic_ntonesmall} to $\log\SNR$ in the limit as the $\SNR$ tends to infinity, we obtain the lower bound
\begin{IEEEeqnarray}{lCl}
\Pi_{R^*_1} &  \geq &  \min(\nr,\ntone)\left(1-\frac{\ntone + \nttwo }{L}\right), \quad L\leq\frac{1}{2\lambda_D}.\IEEEeqnarraynumspace\label{eq:firstbound}
\end{IEEEeqnarray}
The condition $L\leq1/(2\lambda_D)$ is necessary since otherwise \eqref{eq:fadingestimate_error} would not hold. This yields one boundary of the pre-log region presented in Theorem \ref{th:JTD_pre-log}.

\subsection*{\underline{Error Event $m_1'=1$, $m_2' \neq 1$}}

This follows from the proof for the error event $m'_1\neq 1, m'_2=1$ by replacing user 1 by user 2. We thus have
\begin{IEEEeqnarray}{lCl}
\Pi_{R^*_2} &  \geq &  \min(\nr,\nttwo)\left(1-\frac{\ntone + \nttwo }{L}\right), \quad L\leq\frac{1}{2\lambda_D}\IEEEeqnarraynumspace\label{eq:secondbound}
\end{IEEEeqnarray}
yielding the second boundary of the pre-log region presented in Theorem~\ref{th:JTD_pre-log}.


\subsection*{\underline{Error Event $m_1' \neq 1, m_2' \neq 1$}}

As above, the GMI corresponding to the event $\mathcal{E}(m'_1, m'_2)$, $m_1' \neq 1, m_2' \neq 1$ can be evaluated as \cite{IEEE:lapidoth:fadingchannels_howperfect, IEEE:weingarten:gaussiancodes}
\begin{equation}
\label{eq:newequation}
 \gmi_{1+2} (\SNR)  = \sup_{\theta \leq 0} \left(\theta F (\SNR) - \kappa_{1,2} (\theta, \SNR) \right)
\end{equation}
where $F(\SNR)$ is given in \eqref{eq:expectation-F}, and where $\kappa_{1,2} (\theta, \SNR)$, given in \eqref{eq:kappa_1_2} on the top of this page, is the conditional log moment-generating function of the metric $D (m_1', m_2')$ associated with $m_1' \neq 1, m_2' \neq 1$, conditioned on the channel outputs and on the fading estimates. 
\addtocounter{equation}{1}

The $\gmi_{1+2} (\SNR)$ can be viewed as the GMI of an $\nr \times (\ntone + \nttwo)$-dimensional MIMO channel with channel matrix $\left(\HRM_{1,k}, \HRM_{2,k} \right)$. Noting that the channel estimator produces the channel-estimate matrix $\left(\hat \HRM_{1,k}^{(T)}, \hat \HRM_{2,k}^{(T)}\right)$, it thus follows from \cite{IEEE:asyhari_etal:nearest_neighbour_ISIT11} that the pre-log
\begin{equation}
\Pi_{R^*_{1+2}} \triangleq \lim_{\SNR \rightarrow \infty} \frac{I^{\rm gmi}_{1+2} (\SNR) }{\log \SNR}
\end{equation}
is lower-bounded by
\begin{IEEEeqnarray}{rCll}
 \Pi_{R_{1+2}^*} & \geq& \min \left( \ntone + \nttwo, \nr \right)  & \left(1 - \frac{\ntone + \nttwo}{L} \right)\label{eq:thirdbound}
\end{IEEEeqnarray}
for $L\leq1/(2\lambda_D)$. This yields the third boundary of the pre-log region presented in Theorem \ref{th:JTD_pre-log}. 

Combining \eqref{eq:firstbound}, \eqref{eq:secondbound} and \eqref{eq:thirdbound}, and noting that the boundary is maximised for $L$ being the largest integer satisfying $L \leq \frac{1}{2\lambda_D}$, proves Theorem~\ref{th:JTD_pre-log}.



\end{document}